\documentclass[sigconf]{acmart}

\copyrightyear{2023}
\acmYear{2023}
\setcopyright{rightsretained}
\acmConference[ISSAC 2023]{International Symposium on Symbolic and Algebraic Computation 2023}{July 24--27, 2023}{Tromsø, Norway}
\acmBooktitle{International Symposium on Symbolic and Algebraic Computation 2023 (ISSAC 2023), July 24--27, 2023, Tromsø, Norway}\acmDOI{10.1145/3597066.3597110}
\acmISBN{979-8-4007-0039-2/23/07}

\usepackage[utf8]{inputenc}
\usepackage[T1]{fontenc}

\usepackage[english]{babel}

\usepackage{amsmath, mathrsfs}
\usepackage[all]{xy}
\usepackage{stmaryrd}

\usepackage{mathtools}
\mathtoolsset{showonlyrefs,showmanualtags}

\usepackage{tikz}
\usetikzlibrary[patterns]
\usepackage{hyperref}

\usepackage{algorithm}
\usepackage[noend]{algorithmic}

\usepackage{multirow}

\usepackage{amsthm}

\usepackage[shortlabels,inline]{enumitem}

\theoremstyle{plain}
\newtheorem{lem}{Lemma}[section]
\newtheorem{prop}[lem]{Proposition}
\newtheorem{thm}[lem]{Theorem}

\newtheorem{conj}[lem]{Conjecture}

\theoremstyle{definition}
\newtheorem{defn}[lem]{Definition}
\newtheorem{notn}[lem]{Notation}

\theoremstyle{remark}

\newtheorem{rmk}[lem]{Remark}
\newtheorem{expl}[lem]{Example}

\newcommand{\N}{\mathbb{N}}
\newcommand{\NN}{\mathbb N}

\newcommand{\Q}{\mathbb{Q}}
\newcommand{\PP}{\mathcal{P}}
\newcommand{\QQ}{\mathbb{Q}}

\newcommand{\RR}{\mathbb{R}}

\newcommand{\val}{\mathrm{val}}

\newcommand{\TestUAGB}{\mathtt{TestUAGB}}

\DeclareMathOperator{\LT}{LT}
\DeclareMathOperator{\LC}{LC}
\DeclareMathOperator{\LM}{LM}
\DeclareMathOperator{\LTE}{LTE}
\DeclareMathOperator{\Ecart}{\textup{\textrm{Écart}}}

\DeclareMathOperator{\Supp}{Supp}
\DeclareMathOperator{\spoly}{S-Poly}
\DeclareMathOperator{\init}{in}

\newcommand{\Terms}{\mathrm{Terms}}
\newcommand{\X}{\mathbf{X}}

\renewcommand{\i}{\mathbf{i}}
\renewcommand{\j}{\mathbf{j}}

\renewcommand{\r}{\mathbf{r}}
\newcommand{\s}{\mathbf{s}}
\renewcommand{\u}{\mathbf{u}}
\newcommand{\w}{\mathbf{w}}

\DeclareMathOperator{\Vtrop}{\mathrm{trop}}
\newcommand{\trop}{\Vtrop}

\newcommand{\ifnonempty}[3]{%
  \def\tempa{}%
  \def\tempb{#1}%
  \ifx\tempa\tempb 
  #3 
  \else            
  #2
  \fi}

\newcommand{\KX}[1][]{K \{ \X \ifnonempty{#1}{; #1}{} \}}

\newcommand{\Ts}{\mathscr T}
\newcommand{\Ns}{\mathscr N}

\newcommand{\KTX}{K[\![T]\!][\X]}

\begin{document}

\fancyhead{}

\title{Universal Analytic Gröbner Bases and Tropical Geometry}

%
%
\author{Tristan Vaccon}
\affiliation{Universit\'e de Limoges;
 \institution{CNRS, XLIM UMR 7252}
 \city{Limoges}
 \country{France}  
 \postcode{87060}  
}
\email{tristan.vaccon@unilim.fr}
\author{Thibaut Verron}
\affiliation{Johannes Kepler University, 
 \institution{Institute for Algebra}
 \city{Linz}
 \country{Austria}  
}
\email{thibaut.verron@jku.at}

\thanks{
T.~Verron was supported by the Austrian FWF grant P34872.}

\begin{abstract}
  A universal analytic Gröbner basis (UAGB) of an ideal of a Tate algebra is a set containing a local Gröbner basis for all suitable convergence radii.
  In a previous article, the authors proved the existence of finite UAGB's for polynomial ideals, leaving open the question of how to compute them.
  In this paper, we provide an algorithm computing a UAGB for a given polynomial ideal, by traversing the Gröbner fan of the ideal.
  As an application, it offers a new point of view on algorithms for computing tropical varieties of homogeneous polynomial ideals, which typically rely on lifting the computations to an algebra of power series.

  Motivated by effective computations in tropical analytic geometry, we also examine local bases for more general convergence conditions, constraining the radii to a convex polyhedron.
  In this setting, we provide an algorithm to compute local Gröbner bases
  and discuss obstacles towards proving the existence of finite UAGBs.
\end{abstract}





\maketitle

\clubpenalty=1000
\widowpenalty = 1000
\linepenalty=1000
\addtolength{\textfloatsep}{-0.45cm} 

\section{Introduction}
\label{sec:introduction}

The notion of Tate algebras has been introduced
by Tate in \cite{Tate}
to develop analytic geometry over the $p$-adics, founding 
what is now
called rigid geometry. 
This theory has proved to be central to many 
developments in number theory.
In this context, Tate algebras and ideals in Tate algebras serve the same purpose as polynomial
algebras and polynomial ideals in classical algebraic geometry.
Tate algebras are defined as algebras of power series
over a complete discrete valuation field with convergence
conditions such as converging on a given ball
or a polydisk with given radii.

In previous works \cite{CVV,CVV2,CVV3}, the authors showed
that it is possible to define and compute Gröbner bases (GB)
of ideals in Tate algebras, and modern algorithms
for Gröbner bases computations like signature-based algorithms
\cite{CVV2} and FGLM \cite{CVV3} can be adapted to this
setting.

In \cite{CVV3,CVV4}, the authors paved the way
for computations in Tate algebras in case of
overconvergence, \textit{e.g.} ideals
defined by series converging on a bigger polydisk.
Motivated by the application of analytic geometry
in algebraic geometry, an extreme example of this phenomenon is
that of ideals defined by polynomials in a Tate algebra.
In \cite{CVV4}, it was proved that it is possible
to compute a GB of a polynomial ideal 
in a Tate algebra that
is made of polynomials. 
It was also proved that
for any polynomial ideal, 
there exists a universal analytic Gröbner basis
(UAGB), \textit{i.e.} a finite list of polynomials
such that whenever they are seen as converging
power series in a Tate algebra, they form
a Gröbner basis of the corresponding ideal
in this algebra.
Such a UAGB of a polynomial ideal then contains abundant information
on the local behavior of the ideal.
In this paper, we prove that a UAGB is also universal irrespectively of the order used as tie-break in the algebra.
Furthermore, we provide an algorithm to compute a UAGB in finite time (Algorithm~\ref{algo:UAGB_computation}, Theorem~\ref{theo:proof_of_UAGB_algorithm}).

From a universal GB, it is natural
to consider computing the tropical
variety of an ideal.
Over a field with valuation $K$
such as $\mathbb{Q}_p$ or $ \mathbb{F}_p (\!(t)\!)$,
the tropical variety $\trop(V)$ of a variety $V$ defined
by an ideal $I$ can be defined 
as the closure of the image of $V$ by the valuation,
or alternatively using conditions
on leading terms of the elements of $I.$
Acting as a combinatorial shadow of $V$,
many information on $V$ can be recovered from $\trop(V)$.
The developments of tropical geometry have been
plentiful. To only name a few: 
enumerative geometry \cite{Mikh2005},
understanding optimization algorithms
\cite{ABGJ} or 
analyzing artificial neural networks
of the ReLU type.

Universal GB can help in the computation of 
$\trop(V)$ using the second definition
of the tropical variety.
In our context, working with Tate algebras 
instead of polynomial rings 
gives rise to \textit{tropical analytic geometry}.
This emerging field has been 
defined in \cite{Rabinoff}, adapting the language
of tropical geometry to the world
of rigid geometry.

In Section \ref{sec:tropical_geometry}, 
we consider the case of tropical varieties of Tate polynomial ideals.
We show that the tropical variety of a polynomial ideal is the union of the tropical varieties of its Tate completions, which allows to compute the tropical variety using universal analytical Gröbner bases and the Gröbner fan.
This matches what was known for ideals in $k[\![T]\!][\X]$, which were used as a lifting target in existing algorithms for computing tropical varieties over valued fields.
As such, we provide a new point of view on those algorithms, allowing them to work directly on the Tate series without lifting.
To the best of our knowledge, this is the first effective application of tropical analytic geometry.

Finally, motivated by going further into
the development of effective computations
in rigid geometry, we aim at
building up the tools for 
effective computations on \textit{affinoid
subdomains}.
Roughly speaking, affinoid subdomains
are constructed using generalizations
of Tate algebras to more general convergence
conditions (\textit{e.g.} converging on an annulus)
and taking quotients by ideals.

In Section \ref{sec:tate_alg_on_polyhedral_subdomain},
we make some steps into this journey by providing
effective computations of local GB in the
special case of some polyhedral subdomains
as defined in \cite{Rabinoff}.
We conclude with some conjectures on UAGB
in this context, along with examples and
comments.

\section{Setting}

\label{sec:Setting}

\subsection{Tate algebras and Gröbner bases}

In this section, we recall the definition of Tate 
algebras and their theory of Gröbner bases (GB for short).
Let $K$ be a field with a valuation $\val$
making it complete.
Let $\pi$ be a uniformizer 
of $K$, that is an element of valuation $1$.
Typical examples of such a setting are $p$-adic fields like
 $K=\QQ_p$ with
 $\pi = p$
or Laurent series fields like $K=\QQ(\!(T)\!)$ with
 $\pi = T$.

For $\r = (r_{1},\dots,r_{n}) \in \QQ^{n}$,
the \emph{Tate algebra} $\KX[\r]$ is defined as
  \begin{equation}
    \label{eq:1}
    \KX[\r] := \left\{ \sum_{\mathbf{i} \in \NN^{n}} a_{\mathbf{i}}\X^{\mathbf{i}} 
     \text{ s.t. }
     a_{\mathbf{i}}\in K \text{ and } 
     \val(a_\i) - \r{\cdot}\i \xrightarrow[|\i| \rightarrow +\infty]{} +\infty
    \right\}
  \end{equation}
We call the tuple $\mathbf{r}$ the convergence log-radii of the 
Tate algebra.
We define the Gauss valuation of a term 
$a_{\mathbf{i}}\X^{\mathbf{i}}$ as 
$\val_\r(a_{\mathbf{i}}\X^{\mathbf{i}}) = \val(a_\i) - \r{\cdot}\i$, and 
the Gauss valuation of $\sum a_{\mathbf{i}}\mathbf{X}^{\mathbf{i}} 
\in \KX[\r]$ as the minimum of the Gauss valuations of its terms.
The valuation defines a metric on $\KX[\r]$, for which a sequence $(f_n)_{n \in  \NN} \in \KX[\r]$
converges to zero iff $\val_\r (f_n)\xrightarrow[n \rightarrow +\infty]{} +\infty.$

In this article, we shall frequently need to consider all terms with minimal valuation together.
\begin{defn}
  Let $f \in \sum_{\alpha \in \N^n} c_\alpha \X^\alpha \in \KX[\r]$.
  The $\r$-support of $f$ is
  \begin{equation}
    \label{eq:17}
    \Supp_{r}(f) = \left\{  \alpha \textrm{ s.t. } \val_\r (c_\alpha \X^\alpha) = \val_\r (f) \right\},
\end{equation}
  and 
  the  \emph{initial part} of $f$ is
  \begin{equation}
    \init_{\r}(f) = \sum_{\alpha \in \Supp_{\r}(f)} c_{\alpha}\X^{\alpha}.
  \end{equation}
\end{defn}
By definition, for $f \in \KX[\r]$, 
$\init_{\r}(f)$ is a polynomial.

We fix a classical \emph{monomial order} $\leq_m$ on the set of
monomials $\X^{\mathbf{i}}$, which will be used for tie-breaks. 
Given two terms $a \X^\i$ and $b \X^{\mathbf{j}}$  (with $a,b \in 
K^\times$), we write $a \X^\i <_{\r,m} b \X^{\mathbf{j}}$ if
$\val_\r(a\X^\i) > \val_\r(b\X^{\mathbf{j}})$, or
$\val_\r(a\X^\i) = \val_\r(b\X^{\mathbf{j}})$ 
and $\X^{\mathbf{i}} <_m \X^{\mathbf{j}}$.
By definition, the leading term of a Tate series $f=\sum 
a_{\mathbf{i}}\mathbf{X}^{\mathbf{i}} \in \KX[\r]$ is
its maximal term, and is denoted by
$\LT_{\r,m} (f).$ Its coefficient and its monomial
are denoted $\LC_{\r,m}(f)$ and $\LM_{\r,m}(f)$, with $\LT_{\r,m} (f)=\LC_{\r,m}(f)\times \LM_{\r,m}(f).$
For $f,g \in \KX[\r]$, we define their S-polynomial as
\[\spoly(f,g) = \frac{\LT_{{\r,m}}(g)}{D(f,g)}f - 
\frac{\LT_{{\r,m}}(f)}{D(f,g)}g \]
where $D(f,g) = \gcd(\LT_{\r,m}(f),\LT_{\r,m}(g))$.

A Gröbner basis (or GB for short) of an ideal $I$ of $\KX[\r]$ is a set $G \subseteq I$
such that for all $f \in I$, there exists an index
$g \in G$ such that $\LT_{\r,m}(g)$ divides $\LT_{\r,m}(f)$.
A finite Gröbner basis $(g_1, \ldots, g_s)$ is \emph{reduced} if all $\LT_{\r,m}(g_i)$'s are monic, minimally generate $\LT_{\r,m}(I)$ and, 
for any $i$, $\LT_{\r,m}(g_i)$ is the only term of $g_i$ in $\LT_{\r,m}(I)$.

The following theorem was proved in~\cite{CVV}.
\begin{thm}
\label{theo:GB}
Let $I$ be an ideal of $\KX[\r]$, then $I$ admits a finite Gröbner basis.
\end{thm}

%

\subsection{Local Gröbner bases of polynomial ideals in Tate algebras}
\label{sec:polyn-overc-ideals}

For a polynomial ideal $I \subset K[\X]$ and a system of convergence log-radii $\r$,
we define $I_\r$ to be the ideal of $\KX[\r]$ generated by
the polynomials of $I.$ It is the completion of $I$
with respect to $\val_\r$.
The ideal $I_\r$ usually contains many series and polynomials 
not in $I$.
However, as $I$ is dense in $I_\r$,
it was proved in \cite{CVV4} that 
$I_\r$ does not contain more leading terms
than $I$, and that $I_{\r}$ admits a Gröbner basis comprised of polynomials.

\begin{defn}
  Let $I \subset K[\X]$ be an ideal, $\r$ a system of log-radii and $I_{\r}$ the completion of $I$ in $\KX[\r]$.
  An \emph{$\r$-local Gröbner basis} of $I$ is a Gröbner basis of $I_{\r}$ comprised only of polynomials.
\end{defn}

If one needs to vary the convergence log-radii,
the following object is of interest:

\begin{defn}
Let $I \subset K[\X]$ be 
an ideal.
A finite set $G \subset I\subset K[\X]$
such that  for any $\r \in \Q^n,$
$G$ is an $\r$-local GB of $I$
is called  
a \textit{universal analytic Gröbner basis}
of $I$ (UAGB for short).
\end{defn}

In the usual setting, it is required that universal Gröbner bases  be a Gröbner basis for all monomial orders.
Here, the definition requires only that all convergence radii be covered, without any restriction on the tie-breaking monomial order.
However, we prove in Lemma~\ref{lem:val_r_distinct_gb} that given a polynomial ideal, any term ordering can be achieved with a suitable choice of a system of convergence radii.
Finally \cite{CVV4} culminated with the following result:
\begin{thm}
Let~ $I \subset K[\X]$ be 
an ideal.
Then there exists a universal analytic Gröbner basis
of $I$. \label{thm:existence_universal_analytic_GB}
\end{thm}

While the proof was not constructive,
we provide in this article an algorithm
to compute a UAGB of any polynomial ideal.

\subsection{Homogenization and dehomogenization}

Our algorithm to compute a UAGB of a polynomial
ideal will rely on homogenization and
dehomogenization.
We consign here notations and basic properties
taken from \cite[\S 3.3]{CVV4}

\begin{defn}
Let $(\cdot)^*$ and $(\cdot)_*$ be the homogenization
and dehomogenization applications between
$K[\X]$ and $K[\X,t].$
If $I \subset K[\X]$ is an ideal,
we define $I^* \subset K[\X,t]$ to be the
ideal spanned by the $f^*$ for $f \in I$.
\end{defn}

Given $\r \in \QQ^n$ and $\leq_{m}$ a monomial order, we extend the term order
$<_{\r,m}$ to $K[\X,t]$ and $K \left\lbrace \X,t ; \r,0 \right\rbrace$
as follows.

\begin{defn}
Given two terms $a \X^\alpha t^u$
and $b X^\beta t^v$, we write that
$a \X^\alpha t^u <_{(\r,0),m} b X^\beta t^v$
if one of the following holds:
\begin{itemize}[leftmargin=*]
\item $\val_\r(a\X^\i) > \val_\r(b\X^{\mathbf{j}})$ (\emph{i.e.}, $\val_{\r,0}(a \X^\alpha t^u) > \val_{\r,0}(b X^\beta t^v)$).
\item $\val_\r(a\X^\i) = \val_\r(b\X^{\mathbf{j}})$ 
and $\deg (\X^\alpha t^u) < \deg (X^\beta t^v)$.
\item $\val_\r(a\X^\i) = \val_\r(b\X^{\mathbf{j}}),$ $\deg (\X^\alpha t^u) = \deg (X^\beta t^v)$
and $\X^{\alpha} <_m \X^{\beta}.$
\end{itemize}
This defines a term order on $K \left\lbrace \X,t ; \r,0 \right\rbrace$. \label{defn:term_order_on_homogenization}
\end{defn}

With this order, dehomogenization preserves leading
terms of homogeneous polynomials of $K[\X,t].$

\begin{lem}[Lem.~3.5 in \cite{CVV4}]
Let $\r \in \Q^n.$
Let $h \in K[\X,t]$ be a homogeneous polynomial.
Then $\LT_{(\r,0),m}(h)_*=\LT_{\r,m}(h_*).$
Let $f \in K[\X],$
then $\LT_{\r,m}(f)=(\LT_{(\r,0),m}(f^*))_*.$
\label{lem:dehomogenization_and_LT}
\end{lem}

\section{Term orders}
\label{sec:univ-analyt-gb}

\subsection{Convergence radii and term orders}

In this section, we collect different results regarding term orders in Tate algebras.
First, we consider the relation between term orders and systems of convergence log-radii, and show that given finite data (e.g. a finite set of polynomials or an ideal), it is always possible to realize a term order by a suitable choice of system of convergence log-radii.

\renewcommand{\Ts}{F}
\begin{defn}
  Let $\Ts = (f_1,\dots,f_s) \in K[\X]^s.$
  Given a term order $<$, we define
  \begin{equation}
    \label{eq:21}
    \LT_{<}(\Ts) = \{\LT_{<} (f) : f \in \Ts\}.
  \end{equation}
We say that two term orders $<_{1}$ and $<_{2}$ on $K[\X]$ are equivalent with
respect to $\Ts$
if
\begin{equation}
  \label{eq:19}
  \LT_{<_{1}}(\Ts) = \LT_{<_{2}} (\Ts).
\end{equation}
Let $I \subseteq K[\X]^{s}$ be an ideal, we say that two term orders $<_{1}$ and $<_{2}$ on $K[\X]$ are equivalent with respect to $I$ if
\begin{equation}
  \label{eq:20}
   \left\lbrace \LT_{<_{1}} (f) : f \in I \right\rbrace = \left\lbrace \LT_{<_{2}} (f) : f \in I \right\rbrace
\end{equation}
or equivalently
\begin{equation}
  \label{eq:14}
  \LT_{<_{1}}(I) = \LT_{<_{2}}(I).
\end{equation}
\end{defn}

The two relations are connected as follows.
\begin{lem}\label{lem:equiv_gb_equiv_I}
  Let $I \subseteq K[\X]^{s}$, and $<_{1},<_{2}$ be two term orders.
  Let $G$ be a Gröbner basis of $I$ w.r.t. $<_{1}$ and $<_{2}$.
  If $<_{1}$ and $<_{2}$ are equivalent w.r.t. $G$, then they are equivalent w.r.t. $I$.
\end{lem}
\begin{proof}
  Let $<_{1}$ and $<_{2}$ be two term orders equivalent w.r.t. $G$.
  Since $G$ is a GB of $I$ w.r.t. $<_{1}$ and $<_{2}$, $\LT_{<_{1}}(I)$ is spanned by $\LT_{<_{1}}(G)$, which is equal to $\LT_{<_{2}}(G)$ spanning $\LT_{<_{2}}(I)$.
\end{proof}

The following lemma states that modulo equivalence w.r.t. $\Ts$, it is always possible to choose a term order determined only by the convergence condition.

\begin{lem} \label{lem:on_peut_supposer_les_val_r_distincts}
Let $\r \in \Q^n$ and let $\leq_{\r,m}$ be a term order defined by $\val_{\r}$ and a tie-breaking order $\leq_{m}$.
Let $T$ be a finite set of terms in $K[\X]$.
There is an $\s \in \Q^n$
such that for all $t_{1}, t_{2} \in T$,
\[t_{1} > t_2
\iff \val_\s (t_1) < \val_\s (t_2).\] 
In particular, if $\Ts$ is a finite set in $K[\X]$
, any equivalence class of term orders w.r.t. $\Ts$ 
contains a term order $<_{\s,m}$ such that for all $f \in F$ 
, $\LT_{\s,m}(f) = \init_{\s}(f)$.
\end{lem}
\begin{proof}
Thanks to \cite[Th.~1]{Ost75} (see also \cite[Lem.~1.3.1]{GS93}),
there exists some $\u \in \Q^n$
such that 
for all $c_1\X^{\alpha} \neq c_2 \X^{\beta}$ in $T$, 
\[\X^{\alpha} >_{m} \X^{\beta}
\iff \alpha \cdot \u > \beta \cdot \u.\]

By considering the
finite set of pairs of terms in $T$,
there is a small enough $\varepsilon \in \Q_{>0}$
such that for any $t_1, t_2$ in $T$, 
if $\val_\r (t_1) < \val_\r ( t_2)$
then $\val_{\r-\varepsilon \u} (t_1) < \val_{\r-\varepsilon \u} ( t_2).$

Let $\s = \r-\varepsilon \u.$
Then, for any $t_1=c_{1}X^{\alpha}, t_2=c_2 \X^\beta$
in $T$ such that $t_1 >_{\r,m} t_2$,
one of the following is true:
\begin{itemize}
\item $\val_\r (t_1) < \val_\r ( t_2)$
and therefore
$\val_{\r-\varepsilon \u} (t_1) < \val_{\r-\varepsilon \u} ( t_2)$;
\item $\val_\r (t_1) = \val_\r ( t_2)$
and $\X^\alpha >_{m} \X^\beta$, but then
$ \alpha \cdot \u > \beta \cdot \u$ and 
$\val_{\r-\varepsilon \u} (t_1) < \val_{\r-\varepsilon \u} ( t_2).$
\end{itemize} 
Finally, if $\alpha = \beta$ and $\val (c_1)=\val (c_2)$, then $c_{1}\X^{\alpha} \not> c_{2}\X^{\beta}$ and there is nothing to prove.
Therefore, $\s$ satisfies the claim.

The consequence for equivalence classes w.r.t. $F$ follows, by setting $T = \bigcup_{f \in F}\Supp(f)$.
\end{proof}

\begin{lem}
  \label{lem:val_r_distinct_gb}
  Let $I \subseteq K[\X]$ be an ideal, $<_{\r,m_1}$ be a term order and $G$ a local Gröbner basis of $I$ w.r.t. $<_{\r,m_1}$.
  Then there exists $\s \in \QQ^{n}$ such that, for any tie-break order $<_{m_2}$:
  \begin{itemize}
    \item $G$ is a Gröbner basis of $I$ w.r.t. $<_{\s,m_2}$;
    \item $<_{\r,m_1}$ is equivalent to $<_{\s,m_2}$ w.r.t. $I$;
    \item $\LT_{\s,m_2}(I) = \langle \init_{s}(f) : f \in I \rangle$.
  \end{itemize}
\end{lem}
\begin{proof}

  Let $T$ be the set of all terms which appear in $G$ and in the course of Buchberger's algorithm with weak normal form with $G$ as input w.r.t. the $<_{\r,m_1}$ ordering.
  By Lemma~\ref{lem:on_peut_supposer_les_val_r_distincts}, there exists $\s$ such that all terms in $T$ have distinct $\s$-valuation, compatible with the order $<_{\r,m_1}$.
  Let $<_{m_{2}}$ be a tie-break order.

  Note that when running Buchberger's algorithm with WNF, the radii of convergence are only used for determining leading terms.
  So if we run the algorithm with $G$ as input and for the order $<_{\s,m_2}$, all comparisons will be the same, the exact same terms will appear, and the result will be the same: all the S-polynomials have a weak normal form of $0$ w.r.t. $G$, and so $G$ is a GB of $I$ w.r.t. $<_{\s,m_2}$.

  By Lemma~\ref{lem:equiv_gb_equiv_I}, this, together with the fact that the elements of $G$ have the same leading term for both orders, implies that the term orders are equivalent w.r.t. $I$.
  
  Clearly, since for all $g \in G$, $\LT_{\s,m_2}(g) = \init_{\s}(g)$, $\LT_{\s,m_2}(I) \subseteq \langle \init_{s}(f) : f \in I \rangle$.
  For the reverse inclusion, let $f \in I$, and write $\init_{\s}(f) = \LT_{\s,m_2}(f) + r$.
  Since $G$ is a GB of $I$ w.r.t. $<_{\s,m_2}$, there exists a $g \in G$ and a term $\tau$ such that $\LT_{\s,m_2}(f) = \tau \LT_{\s,m_2}(g)$.
  Since all other terms in $g$ have strictly higher valuation, 
  then either $\val_s(f- \tau g)>\val_s(f)$
  and then $r=0$, or $\val_s(f- \tau g)=\val_s(f)$
  and $\init_\s(f- \tau g) = r = \init_{\s}(f) - \LT_{\s,m_2}(f)$.
  Repeating the process until $r=0$, we get that $\init_{\s}(f)$ is a linear combination of terms divisible by a $\LT_{\s,m_2}(g)$ for some $g$'s, and therefore $\LT_{\s,m_2}(I) = \langle \init_{s}(f) : f \in I \rangle$.
\end{proof}

For such an $\s$, we say that $\s$ \emph{defines a term order} for $I$.
The tie-breaking order $<_{m_{2}}$ becomes irrelevant, and we omit it from the notations.
In the rest of the paper, unless specified otherwise, we will always be considering orders $<_{\s}$ where $\s$ defines a term order.

\subsection{Newton polytopes}
\label{sec:newton-polytopes}

In this section, we generalize known results from the classical case to orderings compatible with the valuation, and relating equivalence classes of term orderings w.r.t. a finite set of polynomials, with data from its Newton polytope. Let $\Ts = (f_1,\dots,f_s) \in K[\X]^s.$

\begin{defn}
For $f \in K[\X]$, we define $\Ns(f)$ to be the convex hull of 
\[ \left\{(\val(c)), \alpha_{1},\dots,\alpha_{n}) \text{ for } c\X^{\alpha} \in \Supp(f) \right\} + \mathbb{R}_+ (1,0,\dots,0).\]
We then define $\Ns ( \Ts)$ to be the Minkowski sum of
the $\Ns (f_i)$'s.
\end{defn}

\begin{rmk}
For $n=s=1$, this coincides with the classical
definition of the Newton polygon (up to a symmetry).
\end{rmk}

\begin{lem}
For the convex polyhedron $\Ns (\Ts),$
$\beta  \in \Ns (\Ts)$ is a vertex if and only if
there is some $U=(1,u_1,\dots,u_n) \in \Q^{n+1}$
such that $\beta \cdot U$
is the unique minimum of the $\alpha \cdot U$'s
for $\alpha \in \Ns (\Ts).$ \label{lem:vertex_and_minimal_scalar_product}
\end{lem}
\begin{proof}
Since $\Ns (\Ts)$
is a convex polyhedron,
$\beta$ is a vertex if and only if
there is some $U=(u_0,u_1,\dots,u_n) \in \Q^{n+1}$
such that $\beta \cdot U$
is the unique minimum of the $\alpha \cdot U$'s
for $\alpha \in \Ns (\Ts).$
Since $\Ns (\Ts)$
is defined from half-lines
of the form
$(\val (c_{i,k}),\alpha_{i,k}^{(1)},\dots,\alpha_{i,k}^{(n)})+\mathbb{R}_+ (1,0,\dots,0)$,
then we can further assume
that the $U$ in the
previous equivalence is
such that $u_{0}>0$
and then by multiplying
by a positive rational,
we can assume that 
$U=(1,u_1,\dots,u_n).$
\end{proof}

\begin{prop}
The vertices of $\Ns ( \Ts)$ are in one-to-one correspondence
with the equivalence classes of term orders
with respect to $\Ts.$ \label{prop:Newton_polytope_and_equivalences_classes_of_term_orders}
\end{prop}
\begin{proof}
For any $i \in \llbracket 1, s \rrbracket$, 
we write $f_i = \sum_{j=1}^{l_i} c_{i,j} \X^{\alpha_{i,j}}.$
Let us define the
following set of
index vectors:
\[\mathbf{J}:= \left\lbrace (j_1,\dots,j_s) \in \N^s : \forall i \in \llbracket 1,s \rrbracket, 1 \leq j_i \leq l_i \right\rbrace. \]
For $\j \in \mathbf{J}$, and $i \in \llbracket 1,s \rrbracket,$  we define $D_{i,j_i} = \llbracket 1, l_{i} \rrbracket \setminus \left\lbrace j_i \right\rbrace$ and
we define
\[ C_\j = \left\lbrace
\r \in \Q^n :
\forall i \in \llbracket 1, s \rrbracket, \forall j \in D_{i,j_i} 
, \: \alpha_{i,j_i} {\,\cdot\,} (1,\r) < 
\alpha_{i,j} {\,\cdot\,} (1,\r) \right\rbrace,\]
so that for any $\r \in C_\j$ and $i \in \llbracket 1, s \rrbracket,$
$\LT_\r (f_i) = c_{i,j} \X^{\alpha_{i,\j_i}}.$
Then, thanks to Lemma~\ref{lem:on_peut_supposer_les_val_r_distincts},
 there is one equivalence
class of term orders
with respect to
$\Ts$
for each non empty
$C_\j.$

Being defined as a Minkowski
sum,
$\Ns (\Ts)$
is the convex hull of
the rays $\alpha_\j+\mathbb{R}_+ (1,0,\dots,0)$ for
$\alpha_\j := \sum_{i=1}^s \alpha_{i,j_i}$
and $\j \in \mathbf{J}.$
Thus, its vertices
are among the $\alpha_\j$'s.
Thanks to Lemma~\ref{lem:vertex_and_minimal_scalar_product},
 $\alpha_\j$
is a vertex of $\Ns (\Ts)$
if and only if there
is some 
$U=(1,u_1,\dots,u_n) \in \Q^{n+1}$
such that $\alpha_\j \cdot U$
is the unique minimum of the $\alpha \cdot U$'s
for $\alpha \in \Ns (\Ts).$
Consequently,
if $\j'=\j+(0,\dots,0,j_i',0,\dots,0)-(0,\dots,0,j_i,0,\dots,0)$
only differs from $\j$
on the coordinate $i$ (for some $j_i' \in D_{i,j_i} $),
then we can deduce from
$\alpha_\j \cdot U < \alpha_{\j'} \cdot U$
that 
$\alpha_{i,j_i} \cdot U < \alpha_{i,j_i'} \cdot U.$
It implies that 
$(u_1,\dots,u_n) \in C_\j.$
And the converse is true:
if $(u_1,\dots,u_n) \in C_\j$ then
$\alpha_\j$
is a vertex of 
$\Ns (\Ts).$

This is enough to conclude
that there is a one-to-one  
correspondance between
vertices
of $\Ns (\Ts)$
and equivalence classes of
term orders with respect
to $\Ts.$
\end{proof}

\section{Universal analytic Gröbner bases%
  \protect\footnotemark}
\label{sec:univ-analyt-grobn}

\footnotetext{A toy implementation of the algorithms in this section is available at: \href{https://gist.github.com/TristanVaccon}{https://gist.github.com/TristanVaccon}}


\subsection{Testing whether a set is a UAGB}

The results of Section~\ref{sec:newton-polytopes} are enough
to immediately provide us with a procedure for deciding whether a set is a UAGB (Algorithm~\ref{algo:test_uagb}).
\begin{prop}
$\Ts =(f_1,\dots,f_s)$ is a UAGB of 
$I=\left\langle \Ts \right\rangle$
if and only if for any $\r$
in the equivalence classes of term orders with respect
to $\Ts$, $\Ts$ is a GB of $I_\r.$
In particular, Algorithm~\ref{algo:test_uagb} is correct.
\end{prop}


\begin{algorithm}
	\caption{\texttt{TestUAGB}}
	\label{algo:test_uagb}
	\begin{algorithmic}[1]
		\REQUIRE $F = \{f_1,\dots,f_s\} \in K[\X]$ generating $I \subseteq K[\X]$
		\ENSURE \texttt{True} if $F$ is a UAGB of $I$,
                otherwise $(\mathtt{False},\u)$ with $\u \in \QQ^{n}$ such that $F$ is not a $\u$-local GB of $I$
		\STATE Compute $N = \Ns (f_1,\dots,f_s)$;		
		\FOR{$\beta \in \left\lbrace \textrm{vertices of }N \right\rbrace$}
			\STATE Compute $U=(1,\u)$ characterizing $\beta$
				as in Lemma~\ref{lem:vertex_and_minimal_scalar_product} ;
			\STATE Compute $G_{\u},$ a $\u$-local  GB of $I$ ;
			\IF{ $\exists\, g \in G_{\u}$ not reducible modulo $F$ for the order $<_{\u}$}
				\RETURN $(\mathtt{False},\u)$ ;
			\ENDIF
		\ENDFOR
		\RETURN \texttt{True} ;
	\end{algorithmic}
\end{algorithm}

\subsection{Computing a UAGB}
\label{subsec:UAGB_computation}

We now show how to use that procedure to compute a UAGB.
To that end, we recall the following result from \cite{CVV4}.

\begin{thm}{\cite[Thm 7.6]{CVV4}} \label{thm:Terms_of_I_fini}
Let $I \subset K[\X]$ be an ideal.
Then the set $\Terms(I):=\{\LT(I_\r) \text{ for } \r \in \Q^n \}$
is finite. 
\end{thm}

The proof of the previous theorem relies on
the following lemma, which we also need.

\begin{lem} \label{lem:GBR_for_two_term_orders}
  Let $I \subset K[\X]$ be an ideal.
  Let $\leq_1$ and $\leq_2$ be two term orders such that $\LT_{\leq_1}(I)=\LT_{\leq_2}(I)$.
  Let $G \subset I$ be a reduced (local) GB of $I$ w.r.t. $\leq_1$.
  Then $G$ is also a reduced GB of $I$ w.r.t. $\leq_2.$
\end{lem}
\begin{proof}
Let $g \in G.$
Since $G$ is reduced and $\LT_{\leq_1}(I)=\LT_{\leq_2}(I)$, then the only
term of $g$ belonging to $\LT_{\leq_2}(I)$ is $\LT_{\leq_1}(g)$.
Thus $\LT_{\leq_1}(g)=\LT_{\leq_2}(g)$, and $G$ is a reduced GB of $I$ w.r.t. $\leq_2.$
\end{proof}

Unlike in the classical case, it is not in general possible to guarantee that any polynomial ideal admits a reduced local Gröbner basis for any convergence radii.
However, for homogeneous ideals, \cite[Lem.~7.2]{CVV4} guarantees that there is a reduced local GB comprised only of polynomials, which can then be computed using any GB algorithm from \cite{CVV, CVV2, CVV4, Vaccon:2015, Vaccon:2017,  Vaccon:2018, CM}.

The proofs in \cite{CVV4} relied on homogenization and dehomogenization of ideals and were not constructive.
In this paper, we replace computations in the homogenized ideal by computations in the ideal spanned by the homogenization of its generators.

\begin{notn}
Let $F = (f_1,\dots,f_s)  \in K[\X]^{s}$,
we define
\[F^h = (f_1^{*},\dots,f_s^{*}). \]
\end{notn}

In general, $\langle F^{h}\rangle \subsetneq I^*$
but $I=\langle F^h \rangle_*=(I^*)_*.$
By \cite[Lem 7.5]{CVV4},
the dehomogenization of a GB of 
$(I^*)_{(\r,0)}$ made of homogeneous polynomials
of $K[\X,t]$ is a polynomial GB of $I_\r$.
This result is still true for $\langle F^h \rangle$, by the following lemma. 

\begin{lem}
Let $F = (f_1,\dots,f_s)  \in K[\X]^{s}$, $I = \langle F \rangle$ and
$\r \in \Q^n.$
Let $(h_1,\dots,h_s)$ be a finite Gröbner basis of
$\langle F^h \rangle_{(\r,0)} \subset K \left\lbrace \X,t ; \r,0 \right\rbrace$
made of homogeneous polynomials of $\langle F^h \rangle$ (hence in $K[\X,t]$).
Then $(h_{1,*},\dots,h_{s,*})$ is an $\r$-local GB  of $I$. \label{lem:dehomogenization_of_GB}
\end{lem}
\begin{proof}
Firstly, due to being a dehomogenization of homogeneous elements of
$\langle F^h \rangle$,
the $h_{i,*}$'s are in $I$
(it is enough to dehomogenize an homogeneous
combination of the $f_i^*$).

Secondly,
by \cite[Cor 5.4]{CVV4}, it is enough
to check that for any $f \in I$,
$\LT_\r(f)$ is divisible by one of the $\LT_\r(h_{i,*})$'s.

Let $f \in I.$ 
Since $I = \langle F^h \rangle_*$, there is some homogeneous
polynomial $g \in \langle F^h \rangle$ such that
$g_* = f.$
Then $g \in \langle F^h \rangle \subset (\langle F^h \rangle)_{(r,0)}$
so there is some $i$ such that $\LT_{(\r,0)}(h_i)$
divides $\LT_{(\r,0)}(g).$
Then thanks to Lemma~\ref{lem:dehomogenization_and_LT},
$\LT_\r (f)=\LT_{(\r,0)}(g)_*$,
$\LT_\r (h_{i,*})=\LT_{(\r,0)}(h_i)_*$,
and monomial divisibility is preserved by dehomogenization.
So $\LT_\r (h_{i,*})$ divides $\LT_\r (f)$
and the proof is complete.
\end{proof}

\begin{rmk}
One needs to be careful that the dehomogenization of 
a \textit{reduced} Gröbner basis of $\langle F^h \rangle_{(r,0)}$
(made of homogeneous polynomials of $\langle F^h \rangle$)
need not even be \textit{minimal},
one leading monomial may divide another after dehomogenization
(\textit{e.g} $y t^3$ and $x^2yt$).
\end{rmk}

We can now provide an algorithm
for computing UAGBs.

\begin{algorithm}[H]
	\caption{UAGB}
	\label{algo:UAGB_computation}
	\begin{algorithmic}[1]
		\REQUIRE $\Ts \in K[\X]^s,$ generating $I \in K[\X].$
		\ENSURE $G$ a UAGB for $I.$				
		\STATE	$G:=\left\{   f^* \textrm{ for } f \in \Ts \right\}$ (and define $J=\langle F^h \rangle$) ;
		\WHILE{ $\TestUAGB(G)$ is not  \texttt{True}}	
		\STATE $\r :=$ system of log radii such that $G$ is not a GB of $J_\r$ (as produced by $\TestUAGB$);
		\STATE $H:=\mathtt{ReducedGB}(G,\r)$ \label{algoline:reducedGB}; \hfill // Polynomial reduced local GB
		\STATE $G := G \cup H$ ;
		\ENDWHILE
		\RETURN $\{ g_{*} \text{ for } g \in G\}$ ;
	\end{algorithmic}
\end{algorithm}

\begin{thm}
Algorithm \ref{algo:UAGB_computation} is correct
and computes a UAGB in finite time.
Furthermore, if the input polynomials are homogeneous, the UAGB contains a reduced Gröbner basis for all orders.
\label{theo:proof_of_UAGB_algorithm}
\end{thm}
\begin{proof}
Thanks to Theorem~\ref{thm:Terms_of_I_fini}, if $J= \langle F^h \rangle,$
there is an integer $t$ and a finite set of term
orders $\leq_1,\dots,\leq_t$
such that $\Terms(J)$ is given by $\{\LT(J_{\leq_1}),\dots, \LT(J_{\leq_t})\}$.
Let us assume that $G$ contains a reduced GB 
w.r.t. each of the orders
$\leq_1,\dots,\leq_l$ for some $l\leq t.$

Let us assume that $\TestUAGB(G)$ fails because $G$ is not an $\r$-local GB of $I$ for some $\r.$
Then, thanks to Theorem~\ref{thm:Terms_of_I_fini} and Lemma~\ref{lem:GBR_for_two_term_orders}, $\LT(J_{\leq_i}) \neq \LT(J_{\r})$ for any integer $i$, 
$1 \leq i \leq l$
and $\LT(J_{\leq_j}) = \LT(J_{\r})$ for some integer $j$, 
$l < j \leq t$.
Up to renumbering, we may assume that $j=l+1.$
Let $H$ be the reduced GB of $J$ for $\r.$
Then $ G \cup H$
contains a reduced GB 
for the orders
$\leq_1,\dots,\leq_{l+1}.$

We then prove by induction that,
after at most $t$ calls to \texttt{ReducedGB}
and to \texttt{TestUAGB},
the algorithm outputs $G$ such that $G$
contains a reduced GB of $J$ for each of
$\leq_1,\dots,\leq_t$
and hence, is a UAGB of $J$.

Finally, thanks to Lemma~\ref{lem:dehomogenization_of_GB},
 the dehomogenization of $G$ is a UAGB of $I=J_*$.
 If the input polynomials are homogeneous, the homogenization and dehomogenization steps are trivial, and the property that the UAGB contains a reduced GB for all orders is preserved.
\end{proof}

\begin{rmk}
From the proof of Theorem~\ref{theo:proof_of_UAGB_algorithm},
Algorithm \ref{algo:UAGB_computation}
needs at most $\#\Terms(J)$ loops
to compute a UAGB.
Each loop may however cause many GB computations
as it is unclear how the edges of the Newton polytopes
vary along the computation. 
\end{rmk}

\subsection{Examples}

\begin{expl}
Let $F=[x-7y,y-7y^2]$ in $\mathbb{Q}_7[x,y].$
Then in $\mathbb{Q}_7[x,y,t],$ one finds that
$F^h$ is not a GB for the weight
$[0,2,0].$ The corresponding reduced GB
will add the polynomial
$ x^2 - xt$ and $G_h=[x-7y,yt-7y^2,x^2 - xt]$
is then a UAGB for $\left\langle F^h \right\rangle.$
Thus, $G=[x-7y,y-7y^2,x^2 - x]$ is a UAGB of $\left\langle F \right\rangle.$
\end{expl}

\begin{rmk}
No 
finite approximate interreduction $\overline{F}$ of
$[x-7y,y-7y^2]$ is enough to be a UAGB.
\end{rmk}

\section{Tropical geometry}
\label{sec:tropical_geometry}
\subsection{Analytical tropical varieties}

In this section, we show that tropical geometry on Tate polynomial ideals specializes that of classical polynomial ideals.
In particular, the results of Section~\ref{sec:univ-analyt-gb} give us a Tate analogue of the Gröbner fan, and allow us to generalize the results of~\cite{Ren2015,MR2020} for computing tropical varieties in $\KTX$, to any Tate algebra.

First, we recall the classical notions of tropical geometry, and state their natural generalization to Tate algebras.

\begin{defn}
  Let $\mathbf{w}=(w_0,\dots,w_n) \in \RR_{<0}\times \RR^{n}$ be a system of weights.
  For a monomial $m = X_1^{\alpha_{1}}\cdots X_{n}^{\alpha_{n}}$ and $c \in K$, we define its weighted degree 
  \begin{equation}
    \label{eq:16}
    \deg_{\w}(c X_1^{\alpha_{1}}\cdots X_{n}^{\alpha_{n}}) = w_{0}\val(c)  +  \sum_{i=1}^{n} w_{i}\alpha_{i}.
  \end{equation}
  
  For $f \in K[\X]$, let $\deg_{\w}(f) = \max(\deg_{\w}(t) : t \in \Supp(f))$.
  We define the initial form of $f$ as
  \begin{equation}
    \label{eq:5}
    \init_{\w}(f) = \sum \left\{t : t \in \Supp(f), \deg_{\w}(t) = \deg_{\w}(f)\right\}.
  \end{equation}

  Let $I \subset K[\X]$ be an ideal.
  Then $\init_{\w}(I)$, the initial ideal of $I$ with respect to a system of weights $\w$, is the ideal spanned by all $\init_{\w}(f)$ for $f$ in $I$.
  The \emph{tropical variety} associated to $I$ is then defined as
  \begin{equation}
    \label{eq:2}
    \Vtrop(I) = \{\w \in \RR_{<0} \times \RR^{n} : \text{$\init_{\w}(I)$ does not contain a monomial}\}.
  \end{equation}

  We say that the system of weights $\w$ and the system of log-radii $\r$ are \emph{compatible} if
  \(
    \r = -\left( \frac{w_1}{w_0}, \dots, \frac{w_{n}}{w_0} \right).
    \)
  Conversely, given a system of log-radii $\r$, the system of weights $(-1,r_1,\dots,r_{n})$ is compatible with $\r$.
  The definitions above extend naturally to series and ideals in $\KX[\r]$ by restricting to systems of weights which are compatible with $\r$.
\end{defn}

\begin{rmk}
  In particular, $\Vtrop(I_{\r})$ is either empty or a half-line formed of all the systems of weights compatible with $\r$.
\end{rmk}

If the systems of weights $\w$ and the system of log-radii $\r$ are compatible, then for any term $t$,
\begin{equation}
  \label{eq:4}
  \deg_{\w}(t) = w_{0}\val_{\r}(t).
\end{equation}
This implies that $\deg_{\w}$ is a graduation: for any terms $t,t'$, $\deg_{\w}(t+t') \leq \max(\deg_{\w}(t),\deg_{\w}(t'))$ with equality if $\deg_{\w}(t) \neq \deg_{\w}(t')$, and $\deg_{\w}(tt') = \deg_{\w}(t)+\deg_{\w}(t')$.

The main result of this section is the fact that the tropical variety associated to $I$ is the \emph{union} of the tropical varieties of all its completions $I_{\r}$.

\begin{lem}
  Let $\w$ be a system of weights, let $\r$ be the compatible system of convergence radii and let $\leq_{m}$ be a monomial order.
  Let $f \in \KX[\r]$.
  Then:
  \begin{enumerate}
    \item $\init_{\w}(f) = \init_{\r}(f)$, and in particular it is a polynomial;
    \item $ \LT_{\leq_{m}}(\init_{\w}(f)) = \LT_{\r,\leq_{m}}(f)$;
  \end{enumerate}
\end{lem}
\begin{proof}
  By compatibility between the system of weights and the convergence log-radii, $\val_{\r}(a \X^{\i}) = w_{0}^{-1}\deg_{\w}(a \X^{\i})$, and $\val_{\r}(f) = w_0^{-1}\deg_{\w}(f)$.
  So $\init_{\w}(f)$ is the sum of all terms with minimal Gauss valuation in the support of $f$, which is by definition $\init_{\r}(f)$.
  The rest follows from the convergence properties in Tate series and the definition of the term order.
\end{proof}

\begin{thm}
  \label{th:trop-local}
  Let $I \subseteq K[\X]$ be an ideal.
  Let $\w \in \RR_{<0}\times \RR^{n}$ be a system of weights, and let $\r = -(w_{1}/w_{0},\dots,w_{n}/w_0) \in \RR^{n}$ be the compatible system of convergence log-radii.
  Let $I_{\r} \subseteq \KX[\r]$ be the completion of $I$.
  Then
  \begin{equation}
    \label{eq:6}
   \init_{\w}(I) = \init_{\w}(I_{\r}) \cap K[\X]
  \end{equation}
  and in particular, $\w \in \Vtrop(I)$ if and only if $\w \in \Vtrop(I_{\r})$.
  Globally, 
  \begin{equation}
    \label{eq:18}
    \Vtrop(I) = \bigcup_{\s \in \RR^{n}} \Vtrop(I_{\s}).
\end{equation}
\end{thm}

\begin{proof}
  Clearly $\init_{\w}(I) \subseteq \init_{\w}(I_{\r})\cap K[\X]$.
  Conversely, let us first consider $f \in I_{\r}$ and let $f_{1},\dots,f_{k}$ be polynomials generating $I$.

  There exist series $g_{1},\dots,g_{k}$ such that $f = g_{1}f_{1} + \dots + g_{k}f_{k}$.
  Let $d = \val_{\r}(f)$, and 
  write each series $g_{i}$ as $h_{i} + r_{i}$, where $h_{i}$ is the sum of all terms with Gauss valuation at most $d-\val_{\r}(f_{i})$ and $r_{i} = g_{i}-h_{i}$ has Gauss valuation greater than $d-\val_{\r}(f_{i})$.
  By the convergence property, the $h_{i}$'s are polynomials.
  The decomposition of $f$ becomes
  \begin{equation}
    \label{eq:9}
    f = h_{1}f_{1} + \dots + h_{k}f_{k}  + (r_{1}f_{1} + \dots + r_{k}f_{k})
  \end{equation}
  where the latter group consists exclusively of terms with Gauss valuation greater than $d$.
  So none of those terms can appear in the initial form of $f$, and as a consequence,
$
  \init_{\w}(f) = \init_{\w}(h_{1}f_{1} + \dots + h_{k}f_{k}).
  $
  Since the $h_{i}$ are polynomials, $h_{1}f_{1} + \dots + h_{k}f_{k} \in I$, and therefore $\init_{\w}(f) \in \init_{\w}(I)$.

  Now let $h \in \init_{\w}(I) \cap K[\X]$, there exists series $g_1,\dots,g_{l} \in I_{\r}$ and series $q_1,\dots,q_{l}$ such that
  \begin{equation}
    \label{eq:8}
    h = q_1 \init_{\w}(g_1) + \dots + q_{l}\init_{\w}(g_{l}).
  \end{equation}
  From the above, we know that $\init_{\w}(g_i) \in \init_{\w}(I)$ for all $i$.
  Since $h$ is a polynomial, it has a \emph{maximal} Gauss valuation $d$.
  Similarly to before, any term in $q_{i}$ with Gauss valuation greater than $d - \val_{\r}(g_{i})$ cannot appear in $h$, so those terms must add to zero on the right hand side, and we can assume that the cofactors $q_{i}$ are polynomials, and therefore $h \in \init_{\w}(I)$ and
  the second inclusion is proved.

  The rest of the statement follows by definition.
\end{proof}

\subsection{Analytic Gröbner fan}
\label{sec:grobner-fan}

Similarly to the polynomial case, tropical varieties can be computed using the Gröbner fan of the ideal.
In this section, we recall those definitions.

The relation between tropical varieties and Gröbner fans is the same as in the usual case, namely, that initial forms generalize leading terms.

\begin{defn}
  Let $\w \in \RR_{<0}\times \RR^{n}$ be a system of weights.
  Let $\leq$ be a \emph{term} order on $K[\X]$.
  We say that $\leq$ refines $\w$ if for all terms $t_1,t_2$, $\deg_{\w}(t_1) \geq \deg_{\w}(t_2) \implies t_1 \geq t_2$ 
  .
  We say that $\w$ \emph{defines a term order} for a finite set of polynomials or an ideal if $\r=-(w_{1}/w_{0},\dots,w_{n}/w_0)$ defines a term order for that set or ideal.
  
\end{defn}

If $\w$ defines a term order for a finite set $F$, for all $f \in F$ and all term orders $\leq$ refining $\w$, $\LT_{\leq}(f) = \init_{\w}(f)$.
Similarly, if $\w$ defines a term order for an ideal $I$, then for all term orders $\leq$ refining $\w$, $\LT_{\leq}(I) = \init_{\w}(I)$.

As seen in Lemma~\ref{lem:on_peut_supposer_les_val_r_distincts}, for any finite set of polynomials $F$ or any ideal $I$, and for any monomial order, there exists an equivalent monomial order defined by a system of weights.


\begin{defn}
  Let $I \subset K[\X]$ be an ideal.
  Let $\w$ be a system of weights.
  The analytic Gröbner cone $C_{\w}(I)$ associated to $\w$ and $I$ is the set of all systems of weights $\w'$ such that $\init_{\w}(I) = \init_{\w'}(I)$.
  The analytic Gröbner fan of $I$ is the fan given by all the analytic Gröbner cones of $I$.
\end{defn}

\begin{prop}
  If $\w \in \Vtrop(I)$, then $C_{\w}(I) \subset \Vtrop(I)$.
  In particular, the tropical variety associated to $I$ is a subfan of the Gröbner fan of $I$.
\end{prop}
\begin{proof}
  Whether a system of weights lies in the tropical variety only depends on the initial forms, and therefore applies identically to the cone.
\end{proof}


Similar to the classical case, it allows to compute the tropical variety associated to $I$ by traversing the Gröbner fan.
The following properties are transpositions of corresponding facts in the classical setting, and describe the Gröbner fan.

\begin{prop}
  Let $I \subset K[\X]$ be an ideal. Let $\w$ be a system of weights, $\r$ the convergence radii associated to $\w$, and $\leq$ a term ordering refining $\w$. 
  Let $G$ be a reduced $\r$-local Gröbner basis of $I$ (with $\leq$ as tie-break).
  Then:
  \begin{enumerate}
    \item $\init_{\w}(I) = \langle \init_{\w}(g) : g \in G \rangle$;
    \item for any system of weights $\w_1$, $\w_1 \in C_{\w}(I)$ iff for all $g \in G$, $\init_{\w}(g)=\init_{\w_1}(g)$; 
    \item $C_{\w}(I)$ is the relative interior of a polyhedral convex cone; 
    \item the closure of $C_{\w}(I)$ in $\RR_{<0}\times \RR^{n}$ is the union of all cones $C_{\w'}(I)$ with 
       $
      \forall\, f \in I, \init_{\w}(f) = \init_{\w}(\init_{\w'}(f))
$
    (or equivalently
    $\forall\, g \in G, \init_{\w}(g) = \init_{\w}(\init_{\w'}(g))$);
    \item if $\w_1$ and $\w_2$ are two systems of weights
    such that $\overline{C_{\w_1}} \cap \overline{C_{\w_2}} \nsubseteq \{0\} \times \RR^{n}$ and $C_{\w_1} \neq C_{\w_2}$, then there exists $\w_3$ such that $\overline{C_{\w_1}} \cap \overline{C_{\w_2}} = \overline{C_{\w_3}}$, and it is a facet of both cones;
    \item $C_{\w}(I)$ has maximal dimension $n+1$ if and only if $\w$ defines a monomial order.
\end{enumerate}
\end{prop}
\begin{proof}
  For (1), let $f \in I$.
  Let $f_{\r}$ be the canonical image of $f$ in $\KX[\r]$, it lies in $I_{\r}$.
  Since $G$ is an $\r$-local GB of $I$, its embedding $G_{\r}$ in $I_{\r}$ is a Gröbner basis of $I_{\r}$ w.r.t. $\leq$.
  In particular, $f_{\r}$ reduces to $0$ modulo $G_{\r}$, which in particular implies that there exists a finite sequence of reductions
  \begin{equation}
    \label{eq:12}
   \textstyle f_{\r} \to f_{\r,1} = f_{\r} - t_{1}g_{\r,j_{1}} \to \dots \to f_{\r,k} = f_{\r} - \sum_{i=1}^{k} t_{i}g_{\r,j_{i}}
  \end{equation}
  with $g_{\r,j} \in G_{\r}$ for all $j$, 
  such that $\val_{\r}(f_{\r,k}) > \val_{\r}(f_{\r})$ and $\val_{\r}(f_{\r,j}) = \val_{\r}(f_{\r})$ for $j < k$.
  In particular, $\init_{\w}(f_{\r}) = \sum_{i=1}^{k} t_{i} \init_{\w}(g_{\r,j_{i}})$.
  This is a polynomial equality, which translates into $\init_{\w}(f) = \sum_{i=1}^{k} t_{i} \init_{\w}(g_{j_{i}})$.


  
  For (2), let $\w_1$ be a system of weight.
  The $\Leftarrow$ direction follows from (1).
  Assume that $\w_1 \in C_{\w}(I)$, and consider the term order $\leq_{1}$ obtained by first considering the $\w_1$-degree, then breaking ties using $\leq$.
  This term order refines $\w_1$.
  Furthermore, $\LT_{\leq_{1}}(I) = \LT_{\leq_{1}}(\init_{\w_1}(I)) = \LT_{\leq}(\init_{\w_1}(I))$, the last equality coming from the fact that the initial forms are $\w_1$-homogeneous.
  Since $\init_{\w_1}(I)=\init_{\w}(I)$, we finally get that $\LT_{\leq_{1}}(I) = \LT_{\leq}(I)$.
  By Lemma~\ref{lem:GBR_for_two_term_orders}, $G$ is a reduced Gröbner basis w.r.t. $\leq_{1}$.

  From there, the proof is similar to that of~\cite[Prop.~3.1.17]{Ren2015}.
  Let $g \in G$.
  Because $G$ is reduced, $\LT_{\leq_{1}}(g) = \LT_{\leq}(g)$ and in particular this term lies in both the support of $\init_{\w}(g)$ and $\init_{\w_1}(g)$.
  Consider $\init_{\w}(g)-\init_{\w_1}(g)$, it lies in $\init_{\w_1}(I)$ and its leading term w.r.t. $\leq_{1}$ cannot be $\LT_{\leq_{1}}(g)$.
  Again since $G$ is reduced, this implies that $\init_{\w}(g) - \init_{\w_1}(g) = 0$.

  The next three properties are proved similarly to the classical case, by using the initial part conditions to cut half-spaces and define cones: see \cite[Prop.~3.1.19]{Ren2015} for (3), \cite[Cor.~3.1.20]{Ren2015} for (4) and \cite[Prop.~3.1.21]{Ren2015} for (5).

  Finally, for (6), if $\w$ defines a monomial order for $I$, by considering the finite set $\Supp(G)$, there exists a small enough $\epsilon \in \RR_{>0}$ such that for all $\u \in \RR_{<0}\times \RR^{n}$, and for all $t_1,t_2 \in \Supp(G)$, $\deg_{\w_1}(t_1) < \deg_{\w_2}(t_2) \iff \deg_{\w_1 + \epsilon \u}(t_1) < \deg_{\w_2 + \epsilon \u}(t_2)$.
  Therefore $C_{\w}(I)$ contains a ball of radius $\epsilon$, so it is open and it must have maximal dimension.
  Conversely, if $\w$ does not define a monomial order, there exists a $g \in G$ such that $\init_{\w}(g)$ has at least two terms $t_1$ and $t_2$, so $\w$ lies on the hyperplane $H = \{\val_{\w}(t_1) = \val_{\w}(t_2)\}$.
  For all $\w' \in C_{\w}(I)$, $\init_{\w'}(g) = \init_{\w}(g)$, so $C_{\w}(I) \subseteq H$, and it cannot have maximal dimension.
\end{proof}

Since we know that $I$ has finitely many sets of leading terms by Th.~\ref{thm:Terms_of_I_fini}, the analytical Gröbner fan of $I$ has only finitely many cones of maximal dimension, and therefore it is finite.
If furthermore $I$ is homogeneous, then $I$ admits a universal Gröbner basis $G$ containing a reduced Gröbner basis for all orders, computable using Algorithm~\ref{algo:UAGB_computation}.
In that case, by Prop.~\ref{prop:Newton_polytope_and_equivalences_classes_of_term_orders}, the vertices of $\Ns(G)$ are in one-to-one correspondence with the equivalence classes of term orders with respect to $G$, and by Lemma~\ref{lem:equiv_gb_equiv_I}, this allows to compute all the equivalence classes of term orders with respect to $I$.
From there one can compute the cones of maximal dimension in the analytic Gröbner fan of $I$, and finally all cones in the fan.

Then, for each cone $C$ in the fan, one can pick a system of weights $\w$ in $C$. From the property (1), $\init_{\w}(I)$ is generated by the initial parts of elements of $G$.
Finally, we can decide whether $\init_{\w}(I)$ contains a monomial by computing a basis of $(\init_{\w}(I) : (x_1\cdots x_{n})^{\infty})$, using the algorithms in~\cite{CVV4}.

Just like in the classical case, both with and without valuation, this algorithm is not the most efficient, because the Gröbner fan can be significantly larger than the tropical variety.
In \cite{Ren2015} and \cite{MR2020}, better algorithms have been presented for ideals in $\KTX$, and generalized to $p$-adic fields by lifting back to that case.
Those algorithms still traverse the Gröbner fan, but not in an exhaustive way, and rely on Buchberger's algorithm with Mora reductions for computing bases in the cones.

The conclusion of this section is that the key properties of the Gröbner fan in $\KTX$ are shared across polynomial rings over a valued ring or field, by way of Tate completions, and Buchberger's algorithm with weak normal form allows to compute generators in a similar way.
Thus, it offers an alternative point of view on existing algorithms using liftings to reduce the problem to $\KTX$, instead performing the computations directly in the desired ring or field.

We expect that the other optimized algorithms for computing tropical varieties in $\KTX$ similarly transpose to the general setting.

\begin{rmk}
  The requirement that the ideal is homogeneous can be relaxed into requiring that the ideal admits a reduced Gröbner basis for all orders.
  This requirement is also present in the aforementioned works in $\KTX$.
\end{rmk}

\section{Tate Algebras on polyhedral subdomains}

\label{sec:tate_alg_on_polyhedral_subdomain}

As of now, we have studied Tate algebras using
convergence conditions of the following types: (1) on a polydisk defined by log-radii $\s \in \QQ^n$;
(2) convergence everywhere, \textit{i.e.} the algebra $K[\X]$; (3) on all polydisks $\r \leq \s$ (overconvergence).

To build up the tools for rigid geometry
or tropical analytic geometry as in \cite{Rabinoff},
we need more general convergence conditions
such as convergence on an annulus (\textit{e.g.}
converging for all $\X \in \Q_p^n$ such that 
$\forall i, \: a_i \leq \val (x_i) \leq b_i$ for some $a_i,b_i \in \mathbb{R}$)
or converging on a polyhedron.

Following \cite[Definition 6.3]{Rabinoff}
but restricting to the case of power series
instead of Laurent series,
we define the following
ring of functions $\KX[P]$
for $P$ the rational points of a polyhedron.
This corresponds to a special case of the affinoid algebra defined
by a polyhedral subdomain in~\cite{Rabinoff}.

\begin{defn}
Let $\PP \subset \mathbb{R}^n$ be a polyhedron with vertices in $\mathbb{Q}^n$
and admitting only $\sigma=\mathbb{R}_{<0}^n$
as cone of unbounded direction. Let $P= \PP \cap \QQ^n$.
We define the ring
\[\KX[P] := \left\lbrace \sum_{\alpha \in \N^n} a_\alpha \X^\alpha  : 
   \begin{array}{l}
     a_\alpha \in K \\
     \forall \r \in P, \: \val_\r(a_\alpha \X^\alpha) \rightarrow +\infty
   \end{array}
  \right\rbrace. \]
\end{defn}
The series in $\KX[P]$ are exactly the series 
converging on all polydisks with radius given by
the weights which are points of $P$.
We give an example in Figure~\ref{fig:exp_polyhedron}.

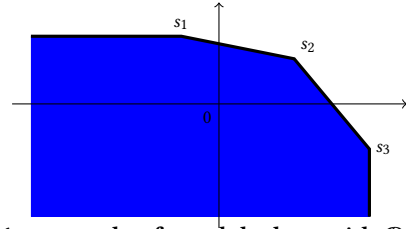
\begin{figure}
\hfill
\begin{tikzpicture}[xscale=0.5, yscale=0.3]
\begin{scope}[fill=blue, opacity=0.2]
  \fill (-5,-5)--(-5,3)--(-1,3)--(2,2)--(4,-2)--(4,-5);
  \fill (-5,-5)--(-5,3)--(-1,3)--(-1,-5);
  \fill (-5,-5)--(-5,2)--(2,2)--(2,-5);
  \fill (-5,-5)--(-5,-2)--(4,-2)--(4,-5);
\end{scope}

\draw[->,thin] (-5.5,0)--(5,0);
\draw[->,thin] (0,-5.5)--(0,4.5);
\draw[very thick] (-5,3)--(-1,3)--(2,2)--(4,-2)--(4,-5);

\node[scale=0.8, below left,thick] at (0,0) { $0$ };
\node[scale=0.8, above] at (-1,3) { $s_1$ };
\node[scale=0.8, above right] at (2,2) { $s_2$ };
\node[scale=0.8, right] at (4,-2) { $s_3$ };
\end{tikzpicture}
\hfill\null

\vspace{-5mm}

\caption{An example of a polyhedron with $\PP$ the convex hull of the $s_i+\RR_{<0}^2$'s in $\RR^2.$}
\label{fig:exp_polyhedron}
\end{figure}

\begin{prop}
If $\PP$ is the convex hull of the $s_1 + \mathbb{R}_{<0}^n,\dots,s_l + \mathbb{R}_{<0}^n$
then
\[ \KX[P] = \bigcap_{i=1}^l \KX[\s_i].\]
\end{prop}
\begin{proof}
The $\subset$ inclusion is clear.
For the converse inclusion, we remark that
if $\r \in P,$
there are some $\lambda_i \geq 0$ such
that $\sum_{i=1}^l \lambda_i =1$
and $\sum_{i=1}^l \lambda_i \s_i \geq \r$.
Consequently, for any term $c_\alpha \X^\alpha \in K[\X]$,
\[ \sum_{i=1}^l \lambda_i \val_{\s_i} (c_\alpha \X^\alpha) \leq \val_\r (c_\alpha \X^\alpha). \]
If $f = \sum_{\alpha \in \N^n} c_\alpha \X^\alpha \in \bigcap_{i=1}^l \KX[\s_i]$ then
for all $i$, $\val_{\s_i} (c_\alpha \X^\alpha) \rightarrow + \infty.$
Thus, $\val_\r (c_\alpha \X^\alpha) \rightarrow + \infty.$
Hence, we get that, if $f \in \bigcap_{i=1}^l \KX[\s_i]$, then
$f \in \KX[P]$ which concludes the proof.
\end{proof}


\subsection{Local Gröbner bases}

In this section, we explain how if $\PP$ is such a polyhedron, $P = \PP \cap \QQ^n$, $I$ an ideal in $\KX[P]$, and $\r \in P $, it is possible to compute an $\r$-local Gröbner basis of $I$ comprised only of elements of $\KX[P]$.
First, we adapt the notion of écarts from \cite[\S 6.1]{CVV4}.

\begin{defn}
  Let $f \in \KX[P]$, and $\r,\s \in P$.
  We define the $(\s,\r)$\mbox{-}degree of $f$ as
 $
    \deg_{\s,\r}(f) = \max_{\alpha \in \Supp_\s(f)} (\s-\r)\cdot \alpha,
$
  and the écarts of first and second type of $f$ as 
  \begin{align*}
    \Ecart_{\s,\r,0}(f)&:=\val_\s(\LT_\r(f))-\val_\s(f), \\
    \Ecart_{\s,\r,1}(f)&:=\deg_{\s,\r}(f)-\deg_{\s,\r}(\LT_\r(f))
  \end{align*}
\end{defn}

\begin{lem}[6.3 in \cite{CVV4}]
$\forall f \in \KX[P],
\forall i = 0,1, 
\Ecart_{\s,\r,i}(f)\geq 0.$
\end{lem}

We adapt Mora's overconvergent Weak Normal
Form algorithm of  \cite{Mora} to computations
in $\KX[\r]$ for series in $\KX[P]$, as in \cite{CVV4}.

\begin{algorithm}[H]
	\caption{$\textsf{WNF}(f,g, P,r)$, Mora's WNF for local computations in polyhedral subdomains}
	\label{algo:Mora_local_over_polyhedra}
	\begin{algorithmic}[1]
		\REQUIRE $f,g_1,\dots,g_s \in \KX[P]$ where $\PP$ is the convex hull of the $s_1+\RR_{<0}^n,\dots,s_t+\RR_{<0}^n$ in $\RR^n$ for some $s_i$'s, $P = \PP \cap \QQ^n$ and $\r \in P.$
		\ENSURE $h \in \KX[P]$ such that:
                \begin{itemize}
                  \item for some $\mu,u_1,\dots,u_s \in \KX[P],$  $\mu f=\sum u_i g_i +h$ 
                  \item either $h = 0$ or $\LT_\r(h)$ is divisible by no $\LT_\r(g_i)$
                  \item $\mu$ is invertible in $\KX[\r]$
                  \item $\LT_\r (u_i g_i) \leq \LT_\r(f)$
                \end{itemize}
						
		\STATE	$h:=f$ ; 
		\STATE  $T:=(g_1,\dots,g_s)$ ;
		\WHILE{ $h \neq 0$ and $T_h := \{g \in T, \LT_\r(g) \mid \LT_\r(h) \} \neq \emptyset$}	
		\STATE take $g \in T_h$ minimizing 
                $\Ecart_{s_1,\r,0}(g)$ first, then $\Ecart_{s_1,\r,1}(g)$, then $\Ecart_{s_2,\r,0}(g)$, \ldots, and finally $\Ecart_{s_t,\r,1}(g)$ ;
		\IF{for any $j,k$, $\Ecart_{s_j,\r,k}(g) >\Ecart_{s_j,\r,k}(h),$}
		\STATE $T:=T \cup \{h\}$;
		\ENDIF
		\STATE $h:=\spoly(h,g)$ ;
		\ENDWHILE
		\RETURN $h$ ;
	\end{algorithmic}
\end{algorithm}

Correctness comes from the following.

\begin{lem}[6.4 in \cite{CVV4}]
If $g \in T_{h_m}$ is such that:
\begin{itemize}
\item $\Ecart_{s_i,\r,0}(g) \leq \Ecart_{s_i,\r,0}(h_m)$,
\item $\Ecart_{s_i,\r,1}(g) \leq \Ecart_{s_i,\r,1}(h_m)$,
\end{itemize}
and if $t=\LT_\r(h_m)/\LT_\r(g)$
and $h_{m+1}=h_m-tg,$ then
\[\val_{s_i}(h_{m+1}) \geq \val_{s_i}(h_m). \]
Moreover, in case of equality, then 
\[ \deg_{s_i,\r}(h_{m+1})\leq \deg_{s_i,\r}(h_m).\]\label{lem:ecarts_vals_degre_s_r}
\end{lem}

\begin{prop}
Algorithm \ref{algo:Mora_local_over_polyhedra}
either terminates in a finite number of steps,
or $\LT_\r(h)$ and all the $\LT_{s_i}(h)$'s converge
to $0.$ \label{prop:terminaison_convergence_WNF_surconvergent}
\end{prop}
\begin{proof}
The proof follows exactly the same lines as 
that of \cite[Prop.~6.5]{CVV4}.
The only difference is the definition of the extended leading terms:
for $h \in \KX[P]$
we define
\begin{equation}
  \label{eq:22}
  \LTE(h):=\prod_{i=1}^t U_i^{d_{i,1} \Ecart_{s_i,\r,0}(h)}\prod_{i=1}^t V_i^{d_{i,2} \Ecart_{s_i,\r,1}(h)}\LT(h)
\end{equation}
in $K[\X,U_1,\dots,U_t,V_1,\dots,V_t]$, for some suitable $d_{i,j} \in \mathbb{N}$
making everything nonnegative integers.
\end{proof}

Using Algorithm \ref{algo:Mora_local_over_polyhedra}
as a replacement for the reduction procedure in Buchberger's algorithm of \cite[Sec.~5]{CVV4}, one can compute
an $\r$-local GB of an ideal of $\KX[P]$.

\subsection{Conjectures}

The present work is only scratching the surface of the study of ideals in $\KX[P]$.
As an example, we conclude with some questions regarding statements which hold in $K[\X]$ or $\KX[\r]$, but are unclear in $\KX[P]$.

\begin{conj}
If~   $I \subset \KX[P]$ is an ideal,
then the set 
$\Terms_P(I)=\left\lbrace \LT_\r(I), \textrm{ for } \r \in P \right\rbrace$ is finite. \label{conj:terms_fini}
\end{conj}

We may remark that this conjecture 
is true for principal ideals.

\begin{prop}
If $f \in \KX[P],$
$\Terms_P(\left\langle f \right\rangle)$ is finite.
\label{lem:finitess_LTr}
\end{prop}
\begin{proof}
Let $s_1,\dots,s_l \in \QQ^n$ be the extremal points of $\PP.$
Then $\PP$ is the convex hull of the $s_i + \mathbb{R}_{<0}^n$
(and $P=\PP \cap \QQ^n)$).
Let $S(f)$ be the set of all $(\val_{\s_1}(c_\alpha \X^\alpha),\dots,\val_{\s_l}(c_\alpha \X^\alpha),\alpha_1,\dots,\alpha_n)$ for $c_\alpha \X^\alpha \in \Supp(f)$.
Up to scalar multiplication of $f$, we can assume that $S(f) \subset \left(\frac{1}{D} \N^l \right) \times \N^n$ 
for some $D \in \N.$
Hence, $S(f)$ only has a finite number of minimal elements
for the product order on $\left(\frac{1}{D} \N^l \right) \times \N^n.$
Let us denote them by $\Sigma.$

Let $\r \in P.$
We prove that if $t=\LT_\r(f)$,
then $t \in \Sigma$.
Let us write $t=c_\alpha \X^\alpha$, and let
$u=c_\beta \X^\beta $ be a term of $f$, belonging to $\Sigma$
and smaller than
$t$ for the product order:
$\val_{\s_i}(u) \leq \val_{\s_i}(t)$ for all $i \leq l$, and
$\beta_i \leq \alpha_i$ for all $i \leq n$.

Since $\r \in P \subset \PP$ and $\PP$ is the convex hull 
of the $s_i + \mathbb{R}_{<0}^n$, 
there are some $\lambda_i \geq 0$ such
that $\sum_{i=1}^l \lambda_i =1$
and $\sum_{i=1}^l \lambda_i \s_i \geq \r$.
Then:
\begin{align*}
\val_\r (u) &= \val (c_\beta)- \r \cdot \beta, \\
			&= \sum_{i=1}^l \lambda_i (\val (c_\beta) - \s_i\cdot \beta)+ \left( \sum_{i=1}^l \lambda_i \s_i -\r \right) \cdot \beta, \\
			&= \sum_{i=1}^l \lambda_i \val_{\s_i} (u)+\left( \sum_{i=1}^l \lambda_i \s_i -\r \right) \cdot \beta. 
\end{align*} 
Since for all $i \leq l,$ $\val_{\s_i}(u) \leq \val_{\s_i}(t)$
and for all $i \leq n$, $\beta_i \leq \alpha_i$
and $\sum_{i=1}^l \lambda_i \s_i \geq \r$,
then 
\[
\val_\r (u) \leq \sum_{i=1}^l \lambda_i \val_{\s_i} (t)+\left( \sum_{i=1}^l \lambda_i \s_i -\r \right) \cdot \alpha
			\leq \val_\r (t).\]

Consequently, $\val_\r (u) \leq \val_\r (t)$
and $\val_\r (t)$ is the minimum of the $\val_\r(v)$
for $v$ a term of $f.$
Hence $\val_\r (u) = \val_\r (t),$
and thus all the above inequalities are equalities.
It then means that $u=t$ and 
$\LT_\r(f) \in \Sigma$, as claimed.
Since $\Sigma$ is finite,
and $\Terms_P(\left\langle f \right\rangle)=\left\langle \left\lbrace \LT_\r (f), \: \r \in P \right\rbrace \right\rangle$ the proposition is proved.
\end{proof}

Our end goal is the following stronger conjecture.
\begin{conj}
If $I \subset \KX[P]$ is an ideal,
then there is a UAGB of $I$ \textit{i.e.}
a finite set $G \subset I$ such
that $G$ is an $\r$-local GB of $I$ for any $\r \in P.$ \label{conj:polyhedral_UAGB}
\end{conj}

We expect zero-dimensional ideals of $\KX[P]$
to be polynomial ideals and the conjecture is clearly true for all polynomial ideals.

\begin{rmk}
The main issue encountered in attempting to prove Conjectures
 \ref{conj:terms_fini} and \ref{conj:polyhedral_UAGB}
 is the fact that for some $\r \in P$,
 there exists a reduced Gröbner basis
 of $I_\r$, but it usually does not
 belong to $I$ and thus is not
 converging for other points of $P.$
 Hence we can not rely on Lemma~\ref{lem:GBR_for_two_term_orders}.
 This issue was overcome for
 polynomials using homogenization
 (Section~\ref{subsec:UAGB_computation} and \cite[Sec.~7]{CVV4}),
 but this technique does not directly
 translate to the case of series.
\end{rmk}

\bibliographystyle{plain}

\end{document}